\documentclass{acm_proc_article-sp}
\usepackage{amsmath}
\usepackage{latexsym}
\usepackage{graphicx}
\usepackage{amssymb}
\usepackage[boxed]{algorithm2e}

\title{New Insights from an Analysis of Social Influence Networks under the Linear Threshold Model}
\numberofauthors{2}

\author{
\alignauthor
Srinivasan Venkatramanan\\
    \affaddr{Department of Electrical Communication Engineering}\\
    \affaddr{Indian Institute of Science}\\
    \affaddr{Bangalore}\\
    \email{vsrini@ece.iisc.ernet.in}
\alignauthor
Anurag Kumar\\
    \affaddr{Department of Electrical Communication Engineering}\\
    \affaddr{Indian Institute of Science}\\
    \affaddr{Bangalore}\\
    \email{anurag@ece.iisc.ernet.in}
}

	\newtheorem{theorem}{Theorem}[section]
	\newtheorem{lemma}{Lemma}[section]
	\newtheorem{sublemma}{Sublemma}[section]
\begin{document}
\maketitle
\begin{abstract}
We study the spread of influence in a social network based on the Linear Threshold model. We derive an analytical expression for evaluating the expected size of the eventual influenced set for a given initial set, using the probability of activation for each node in the social network. We then provide an equivalent interpretation for the influence spread, in terms of acyclic path probabilities in the Markov chain obtained by reversing the edges in the social network influence graph. We use some properties of such acyclic path probabilities to provide an alternate proof for the submodularity of the influence function. We illustrate the usefulness of the analytical expression in estimating the most influential set, in special cases such as the UILT(Uniform Influence Linear Threshold), USLT(Uniform Susceptance Linear Threshold) and node-degree based influence models. We show that the PageRank heuristic is either provably optimal or performs very well in the above models, and explore its limitations in more general cases. Finally, based on the insights obtained from the analytical expressions, we provide an efficient algorithm which approximates the greedy algorithm for the influence maximization problem.
\end{abstract}



\section{Introduction}

A social network models a set of entities (such as individuals or organizations) that are tied by one or more types of interdependency (such as friendship, collaboration or co-authorship). Typically each individual is a node in the social network, and there is an edge between two nodes, if there exists some form of interaction between them. Real world social networks such as scientific collaboration networks, have been observed \cite{newman01scientific-collaboration} to exhibit several properties of complex networks, such as scale-free degree distribution and the small-world phenomenon. Given a social network, there are several well established node-selection heuristics such as degree centrality and distance centrality whose effectiveness have been analysed in \cite{wasserman94social-net-analysis}. In this paper we analyze and derive new insights on the spread of influence under the Linear Threshold model studied by Kempe et al.~\cite{kempe-etal03max-spread-infl}.

\textbf{Related Literature:} Social networks play a fundamental role as a medium for the spread of information, ideas and influence among its members. Network diffusion processes have been investigated extensively in the past, with focus on spread of epidemics, diffusion of innovation and decision models. The concept of using threshold models to explain collective behaviour was first put forward by Granovetter in \cite{granovetter78threshold-models}, where he discusses the spread of binary decisions, among a group of rational agents, for instance in voting models. Similar behaviours can also be observed in cases of innovation adoption, rumour and disease spreading. Newman \cite{newman02spread-disease-nets} studied the spread of disease on networks under the susceptible-infected-removed (SIR) model and showed how concepts from percolation theory can be used to study these models on a wide variety of networks. 
 
Domingos and Richardson \cite{domingos01mining-network-value,domingos02mining-viral-marketing} were the first to study information diffusion under the viral marketing perspective, and they proposed the concept of a customer's network value, apart from his intrinsic value. They were also the first to  pose the combinatorial optimization problem of choosing the initial set of customers to maximize the net profits, and showed that choosing the right set of users for the marketing campaign could make a large difference. Kempe et al.\ \cite{kempe-etal03max-spread-infl} studied the problem of choosing the most influential initial set using two different models of information propagation, namely the Linear Threshold model (LT model) and the Independent Cascade model (IC model), and showed that the problem is NP-hard and the objective function is \emph{submodular}. They proposed a \emph{greedy approximation algorithm} that was shown to achieve an approximation factor of $(1-1/e)$. They also provided generalizations of the two models, and showed how the two generalized models can be made equivalent. 

Web page ranking algorithms such as Google's PageRank \cite{brin-page99pagerank} can also be extended as a heuristic to the social network context, for ranking nodes in order of influence. Kimura et al.\ \cite{kimura06tractable-models} develop upon the Independent Cascade model introduced in \cite{kempe-etal03max-spread-infl} and suggest two special cases of the IC model, which are computationally more efficient, and are good approximations to the IC model when the propagation probabilities are small. Kimura et al.\ \cite{kimura07infl-nodes-bond-percolation} have also used the concept of bond percolation, to easily evaluate the expected influence of a given set of nodes, and hence proposed a faster version of the greedy algorithm. In \cite{leskovec-etal07cost-effective-outbreakdetection} the authors propose a general framework for cost effective outbreak detection, of which the influence maximization problem is a special case, and, by exploiting the submodularity of the influence function, propose the CELF algorithm which achieves close to greedy algorithm performance. Wei Chen et al. \cite{chen-etal09infl-maximization} study the IC model and propose an improved version of the greedy algorithm and also the degree discount heuristic which are found to perform on par with the greedy algorithm.

\textbf{Our Contributions:}We develop upon the Linear Threshold model studied by Kempe et al.\cite{kempe-etal03max-spread-infl}. Our major contributions are as follows:
\begin{itemize}
\item  We derive recursive expressions for the expected influence of a given initial set (in Section~\ref{recur-exp}), provide an interpretation via Acyclic Path Probabilities in Markov chains,and provide an alternate proof of submodularity of the objective function (in Sections~\ref{DTMC-interpret} and \ref{sec-submod-monotone}).

\item We provide some sample cases where the PageRank algorithm is provably optimal or performs very well (in Section~\ref{examples}) and subsequently we discuss the limitations of PageRank in more general cases.

\item We also propose the G1-Sieving algorithm to find the most influential set, based on the insights derived from the recursive expression(in Section~\ref{sec-g1s}) and find that G1-sieving performs almost on par with the Greedy algorithm and is also very efficient in terms of computation. 
\end{itemize}

\section{The Social Network Model}
\label{sec-model}
\subsection*{Glossary of Notation}
$\mathcal{N}$ - weighted directed graph of the entire social network \newline
$\mathcal{N}\backslash \mathcal{A}$ - graph obtained by removing nodes in $\mathcal{A}\subseteq \mathcal{N}$ and all links to or from these nodes\newline
$\mathbf{W}$ - influence matrix with $w_{i,j}$ as entries, gives the edge weights of $\mathcal{N}$ \newline
$\Theta_{j}$ - $U[0,1]$ random threshold chosen by node $j$ \newline
$b_{j}(A) = \sum_{i \in A} w_{i,j}$, total influence into node $j$ from set $A$\newline
$\mathcal{A}_{0}$ - Initial active set \newline
$A_{k}$ - Set of all active nodes at time step $k$, $A_0 \subset A_1 \subset A_2 \ldots$ \newline
$D_{k}$ - Set of nodes which were activated at time step $k$, \( D_{k} = A_{k} \backslash A_{k-1} \) \newline
$S$ - Random time at which the activation process stops, \( S = \min_{k}\{ A_{k} = A_{k-1} \} \) \newline
$g_j^{(\mathcal{N},\mathcal{A})} (k)$ =  $\mathbb{P}^{(\mathcal{N},\mathcal{A})}(j \in D_{k})$  = $\mathbb{P}^{(\mathcal{N})}(j \in D_{k} \big{|}\mathcal{A}_0 = \mathcal{A})$\newline
$g_j^{(\mathcal{N},\mathcal{A})}$ =  $\mathbb{P}^{(\mathcal{N},\mathcal{A})} (j \in A_{S})$ \newline
$\sigma^{(\mathcal{N},\mathcal{A})}$ = $\mathbb{E}^{(\mathcal{N},\mathcal{A})}[|A_S|]$\newline

\paragraph{Social Network Description}
In this work, we adopt a model in which a social network is a weighted directed graph $\mathcal{N}=(V,E)$, where the edge weights $w_{i,j}$ give a measure of influence of node $i$ on node $j$. The activation process begins with an initial set of active nodes $\mathcal{A}_{0}$, and at each step $k$ the set of active nodes $A_{k}$ keeps increasing, due to the influence of the already active nodes. This goes on until a \emph{terminal set} $A_{S}$ is reached, from where the activation process cannot proceed further. We shall focus only on the \emph{progressive case}, where nodes once activated, will never switch back to the inactive state.

\paragraph{Activation Models}
There are two widely used activation models, namely, \emph{Linear Threshold model} and \emph{Independent Cascade model}. In the \emph{Linear Threshold model}, we ensure that \( \sum_{i\neq j}  w_{i,j} \leq 1 \). In this model, each node $j$ randomly chooses a threshold $\Theta_{j}$ uniformly from [0,1] \emph{at the beginning}; At step $k$, a node $j$ gets activated if, it had been inactive until step $k-1$ and
\[ \sum_{i \in A_{k-1}}   w_{i,j} \geq \Theta_{j} \] 
In the \emph{Independent Cascade model}, we start with an initial active set $\mathcal{A}_{0}$ and the activation proceeds according to the following randomized rule. Whenever a node $i$ becomes active at step $k$, it is given one attempt at activating each of its inactive neighbours $j$, succeeding with probability $w_{i,j}$. If $i$ succeeds, $j$ becomes part of $A_{k+1}$, but whether or not $i$ succeeds, it cannot make any more attempts at activating its neighbours in subsequent rounds. Again, the activation process continues until no more activations are possible. Kempe et al. also provide generalizations of the above two models in \cite{kempe-etal03max-spread-infl}, and show how the two generalized models can be made equivalent. In the remaining sections of this paper, we shall be discussing only the Linear Threshold model. 
 
\paragraph{Problem Statement}
Given the initial set $\mathcal{A}_0$, the activation process evolves in discrete time steps according to the Linear Threshold model. Let $A_k$ denote the set of all active nodes at time $k$. Since we are dealing with the progressive case, it is clear that $A_0 \subset A_1 \subset \ldots \subseteq \mathcal{N}$. Let $D_k$ denote the set of nodes which were activated at time $k$, i.e., $D_k = A_k \backslash A_{k-1}$ and $D_0=\mathcal{A}_0$. Let $S$ denote the random stopping time at which the activation process stops, i.e., $S=min_k\{A_k = A_{k-1}\}$. Then we can define $\sigma^{(\mathcal{N},\mathcal{A}_0)} = \mathbb{E}^{(\mathcal{N},\mathcal{A}_0)}[|A_S|]$ to be the expected size of the \emph{terminal set} $A_S$, starting with $\mathcal{A}_0$ as the initial set in the network $\mathcal{N}$.
The influence maximization problem can then be formulated as follows:

\begin{equation}
\max  \sigma^{(\mathcal{N},\mathcal{A}_0)}
\label{eq-obj-fn}
\end{equation}
\[\text{s.t.} \ \mathcal{A}_0 \subset \mathcal{N} \]
\[ |\mathcal{A}_0|= K \]

\paragraph{Greedy Algorithm}
The greedy hill climbing solution for the influence maximization problem is shown in Algorithm~\ref{algo-greedy}. In the algorithm, $K$ is the size of the required initial set, and the set $X$ obtained after the $K$ iterations is the greedy solution. It is noted in \cite{kempe-etal03max-spread-infl} that this achieves an approximation factor of $(1-1/e)$, and the proof involves the submodularity and monotonicity of $\sigma^{(\mathcal{N},\mathcal{A})}$.

\begin{algorithm}
$X \leftarrow \emptyset$\;
\For {$i=1$ to $K$}
{
Choose $v_i$ such that $v_i =\arg \max_v \sigma^{(\mathcal{N}, X \cup v)}$\;
$ X \leftarrow X \cup v_i$\;
}
\caption{Greedy Algorithm}
\label{algo-greedy}
\end{algorithm}  

\section{ Recursive Expression for \large{$\sigma^{(\mathcal{N},\mathcal{A}_0)}$}}
 \label{recur-exp}
 As far as we know, there is no work on mathematically characterising the value of $\sigma^{(\mathcal{N},\mathcal{A}_0)}$ for the models introduced in \cite{kempe-etal03max-spread-infl}. Moreover, $\sigma^{(\mathcal{N},\mathcal{A}_0)}$ is generally obtained by simulating the activation process several times on the social network, and taking the average value. In this section, we derive an expression for $\sigma^{(\mathcal{N},i)}$ in recursive form, and hence give a general expression for $\sigma^{(\mathcal{N},\mathcal{A}_0)}$. We use this expression later to provide insights into various existing heuristics, and also for proposing an efficient algorithm that matches the greedy solution.
Let us begin with the definition of $\sigma^{(\mathcal{N},\mathcal{A}_0)}$.
\[ \sigma^{(\mathcal{N},\mathcal{A}_0)} = \mathbb{E}^{(\mathcal{N},\mathcal{A}_0)}[|A_S|] \]
Note that, since $D_{k}$'s are disjoint, and $\bigcup_{k=0}^{\infty} D_{k} = A_S$, we can write,
\begin{eqnarray*}
\sigma^{(\mathcal{N},\mathcal{A}_0)} &=& \sum_{k=0}^{|\mathcal{N}|}\mathbb{E}^{(\mathcal{N},\mathcal{A}_0)}[|D_k|]\\
&=&\sum_{k=0}^{|\mathcal{N}|}\mathbb{E}^{(\mathcal{N},\mathcal{A}_0)} [ \sum_{j\in \mathcal{N}} I_{\{j \in D_k\}} ]\\
&=&\sum_{k=0}^{|\mathcal{N}|}\sum_{j\in \mathcal{N}} \mathbb{E}^{(\mathcal{N},\mathcal{A}_0)} [ I_{\{j \in D_k\}}]\\
&=&\sum_{k=0}^{|\mathcal{N}|}\sum_{j\in \mathcal{N}} g_j^{(\mathcal{N},\mathcal{A}_0)} (k)\\						
\end{eqnarray*}						 
In the above expressions, $I_{\{E\}}$ denotes the indicator variable for the event $E$, and we also use the fact that the total number of time steps of the activation process is bounded above by the number of nodes in the network, $|\mathcal{N}|$. $g_j^{(\mathcal{N},\mathcal{A}_0)} (k)$ gives the probability that node $j$ is activated at the time step $k$, given that we start with $\mathcal{A}_0$ as the initial set in the network $\mathcal{N}$. We wish to state the following lemma, which will help us determine $g_j^{(\mathcal{N},\mathcal{A}_0)} (k)$.

\begin{lemma}\label{lemma-g}
\begin{enumerate}
\item $j \in \mathcal{A}_0$,
      \begin{enumerate}
      \item $g_j^{(\mathcal{N},\mathcal{A}_0)} (0) =1$
      \item $g_j^{(\mathcal{N},\mathcal{A}_0)} (k) =0$ , for all $k > 0$
	  \end{enumerate}
	  
\hspace{1cm}\item $j \notin \mathcal{A}_0$,	  
      \begin{enumerate}
	  \item $g_j^{(\mathcal{N},\mathcal{A}_0)} (0) =0$
	  \item $g_j^{(\mathcal{N},\mathcal{A}_0)} (k) = \displaystyle{\sum_{l \in \mathcal{N} \backslash \{j\}}} g_l^{(\mathcal{N}\backslash \{j\},\mathcal{A}_0)} (k-1)\ w_{l,j}$            , for all $k > 0$
      \end{enumerate}
\end{enumerate}
\end{lemma}
\begin{proof}
Note that 1(a) and 2(a) are obvious, since $D_0=\mathcal{A}_0$, chosen deterministically. 1(b) follows from 1(a) and the observation that
\( \sum_{k=0}^{\infty} g_{j}^{(\mathcal{N},\mathcal{A}_0)} (k) \leq 1 \) by definition.
For 2(b), since $\mathcal{A}_0 \subset \mathcal{N} \backslash \{j\}$,
\[ g_{j}^{(\mathcal{N},\mathcal{A}_0)} (k) = \mathbb{P}^{(\mathcal{N} \backslash \{j\}, \mathcal{A}_0)} \bigg( b_{j}(A_{k-2}) < \Theta_{j} \leq b_{j}(A_{k-1}) \bigg) \]
Since $D_{k-1}=A_{k-1} \backslash A_{k-2}$, and $\Theta_{j}$ is chosen uniformly from [0,1], we can write,

\begin{eqnarray*}
 g_{j}^{(\mathcal{N},\mathcal{A}_0)} (k) &=& \mathbb{E}^{(\mathcal{N} \backslash \{j\}, \mathcal{A}_0)} [ b_{j}(D_{k-1})]\\
 &=&\mathbb{E}^{(\mathcal{N} \backslash \{j\}, \mathcal{A}_0)} \bigg[ \sum_{l \in \mathcal{N}\backslash\{j\}} I_{\{l \in D_{k-1}\}} w_{l,j} \bigg]\\
 &=&\sum_{l \in \mathcal{N} \backslash \{j\}} g_l^{(\mathcal{N}\backslash \{j\},\mathcal{A}_0)} (k-1)\  w_{l,j}\\
\end{eqnarray*}
\end{proof}

\subsection{Singleton Initial Set}
Now, by Lemma~\ref{lemma-g}, we can write,
For $j \neq i$,
\begin{eqnarray*}
g_{j}^{(\mathcal{N}, i)}(1) &=& w_{i,j}\\
\end{eqnarray*}
\begin{eqnarray*}
g_{j}^{(\mathcal{N}, i)}(2) &=& \sum_{k_1 \in \mathcal{N} \backslash \{j\}} g_{k_1}^{(\mathcal{N}\backslash \{j\},i)}(1) w_{k_1,j}\\
							&=& \sum_{k_1 \in \mathcal{N} \backslash \{i,j\}} w_{i,k_1} w_{k_1,j}\\
\end{eqnarray*}
\begin{eqnarray*}
\lefteqn{g_{j}^{(\mathcal{N}, i)}(3)}\\
&=& \sum_{k_1 \in \mathcal{N} \backslash \{j\}} g_{k_1}^{(\mathcal{N}\backslash \{j\},i)}(2) w_{k_1,j}\\
                            &=& \sum_{k_1 \in \mathcal{N} \backslash \{i,j\}} g_{k_1}^{(\mathcal{N}\backslash \{j\},i)}(2) w_{k_1,j}\\
							&=& \sum_{k_1 \in \mathcal{N} \backslash \{i,j\}} \sum_{k_2 \in \mathcal{N} \backslash \{j,k_1\}} g_{k_2}^{(\mathcal{N}\backslash \{j,k_1\},i)}(1) w_{k_2,k_1} w_{k_1,j}\\
							&=& \sum_{k_1 \in \mathcal{N} \backslash \{i,j\}} \sum_{k_2 \in \mathcal{N} \backslash \{i,j,k_1\}} w_{i,k_2} w_{k_2,k_1} w_{k_1,j}\\
							&=& \sum_{k_1 \in \mathcal{N} \backslash \{i,j\}} \sum_{k_2 \in \mathcal{N} \backslash \{i,j,k_1\}} w_{i,k_1} w_{k_1,k_2} w_{k_2,j}\\
\end{eqnarray*}
In the each of the steps, we substitute the expression for  $g_{j}^{(\mathcal{N},i)} (k)$ and noting that $g_{i}^{(\mathcal{N},i)} (k) =0$ for $k>0$. The last step is obtained by suitably rearranging the terms.

Note that, the above terms can be understood, as the influence of node $i$ reaching node $j$ through a path (without loops) of $k$ hops. We can use this to derive the recursive equation for $\sigma^{(\mathcal{N},i)}$.We have,

\begin{eqnarray*}
\lefteqn{\sigma^{(\mathcal{N},i)}}\\ 
&=& \sum_{k=0}^{\infty} \sum_{j \in \mathcal{N}} g_j^{(\mathcal{N},i)} (k)\\
&=&1 + \sum_{k=1}^{\infty} \sum_{j \in \mathcal{N}} g_j^{(\mathcal{N},i)} (k) \\
&=& 1 + \sum_{j \in \mathcal{N} \backslash \{i\}} g_j^{(\mathcal{N},i)} (1) + \sum_{j \in \mathcal{N} \backslash \{i\}} g_j^{(\mathcal{N},i)} (2) + \cdots \\
&=& 1 + \sum_{j \in \mathcal{N} \backslash \{i\}} w_{i,j} + \sum_{j \in \mathcal{N} \backslash \{i\}}  \sum_{k_1 \in \mathcal{N} \backslash \{i,j\}} w_{i,k_1} w_{k_1,j} + \cdots \\
\end{eqnarray*}
By changing variables and rearranging summations, this is equivalent to
\begin{eqnarray*}
\sigma^{(\mathcal{N},i)} &=& 1 + \sum_{k_1 \in \mathcal{N} \backslash \{i\}} w_{i,k_1} +\\
& & \sum_{k_1 \in \mathcal{N} \backslash \{i\}} w_{i,k_1} \sum_{j \in \mathcal{N} \backslash \{i,k_1\}} w_{k_1,j} + \cdots\\
&=& 1 + \sum_{k_1 \in \mathcal{N} \backslash \{i\}} w_{i,k_1} \bigg[ 1 + \sum_{k_2 \in \mathcal{N} \backslash \{i,k_1\}} w_{k_1,k_2} \bigg[ 1 + \cdots\\
\end{eqnarray*}
Note that this equation is recursive in nature, and hence we can state the following theorem.

\begin{theorem}
Given a social network $\mathcal{N}$, with influence matrix $\mathbf{W}$, the total influence of any node $i$ in the network under the LT model is given by 
\begin{equation}
\sigma^{(\mathcal{N},i)} = 1 + \sum_{j \in \mathcal{N} \backslash \{i\}} w_{i,j} \sigma^{(\mathcal{N} \backslash \{i\},j)}
\label{eq-sigma-i}
\end{equation}
\end{theorem} 

The equation says that under the linear threshold model, the total influence of any node $i$ in the network, is one (for the node $i$ itself) plus the weighted sum of the influences of its neighbours in the network without $i$.

\subsection{Initial Set $\mathcal{A}_0$}
A similar derivation can be done for any $\mathcal{A}_0$. In this case, again using Lemma 1, we get

\begin{eqnarray*}
g_{j}^{(\mathcal{N}, \mathcal{A}_0)}(1) &=& \sum_{i \in \mathcal{A}_0} w_{i,j}\\
g_{j}^{(\mathcal{N}, \mathcal{A}_0)}(2) &=& \sum_{i \in \mathcal{A}_0} \sum_{k_1 \neq j \  k_1 \notin \mathcal{A}_0} w_{i,k_1} w_{k_1,j}\\
g_{j}^{(\mathcal{N}, \mathcal{A}_0)}(3) &=& \sum_{i \in \mathcal{A}_0} \sum_{k_1 \neq j \  k_1 \notin \mathcal{A}_0} \sum_{k_2 \neq j,k_1 \  k_2 \notin \mathcal{A}_0} w_{i,k_1} w_{k_1,k_2} w_{k_2,j}\\
\end{eqnarray*}
and so on. Note that these terms can be understood, as the influence of nodes $i\in \mathcal{A}_0$ reaching node $j$ through a path (without loops) of $k$ steps, \emph{without passing through any other node in $\mathcal{A}_0$}.

Also, having chosen $\mathcal{A}_0$, the edge weights $ \{ w_{i,j} , j \in \mathcal{A}_0 \}$ do not have any effect on $\sigma^{(\mathcal{N},\mathcal{A}_0)}$. By the above two observations, we can thus divide the problem of finding $\sigma^{(\mathcal{N},\mathcal{A}_0)}$ into $K$ subproblems, where $K=|\mathcal{A}_0|$. 
Define sub-networks $\mathcal{N}_{i}^{\mathcal{A}_0}$, for all $i \in \mathcal{A}_0$, such that,

\[\mathcal{N}_{i}^{\mathcal{A}_0} = \{ \mathcal{N} \backslash \mathcal{A}_0 \} \cup \{i\} \]
Then we can see that, 
\[g_{j}^{(\mathcal{N}, \mathcal{A}_0)}(k) = \sum_{i \in \mathcal{A}_0} g_{j}^{(\mathcal{N}_{i}^{\mathcal{A}_0}, i)}(k) \]
Now we can state the theorem as follows.

\begin{theorem}
Given a social network $\mathcal{N}$, with influence matrix $\mathbf{W}$, the total influence of any initial set $\mathcal{A}_0$ in the network under the LT model is given by 

\begin{equation}
\sigma^{(\mathcal{N},\mathcal{A}_0)} = \sum_{i \in \mathcal{A}_0} \sigma^{(\mathcal{N}_{i}^{\mathcal{A}_0},i)}
\label{eq-sigma-a0}
\end{equation}
\end{theorem}

Each of the terms in the right hand side, can then be evaluated recursively using Equation~\ref{eq-sigma-i}.
 
\section{Interpretation via Acyclic Path Probabilities in a DTMC}
\label{DTMC-interpret}
To begin with, since $\sum_{i \neq j} w_{i,j} \leq 1$, $\mathbf{W}$ in general need not be a stochastic matrix. In order to interpret the expressions for $g_{j}^{(\mathcal{N},\mathcal{A}_0)}$ in the Discrete Time Markov chains(DTMC) framework, we require $\mathbf{P}=\mathbf{W}^T$ to be row stochastic. Hence, we can set $w_{j,j} = 1 - \sum_{i \neq j} w_{i,j}$. Note that this does not affect the theory developed till now, since terms of the form $w_{i,i}$ do not feature in Equations~\ref{eq-sigma-i} and \ref{eq-sigma-a0}. We shall also adopt a similar approach when we use the PageRank algorithm, where we will be calculating the stationary probability of a DTMC. 

From Lemma~\ref{lemma-g}, for all $j \neq i$,

\begin{eqnarray*}
\lefteqn{g_{j}^{(\mathcal{N},i)}(k)}\\
&=& \sum_{l_1 \neq i,j} \sum_{l_2 \neq i,j,l_1} \ldots \sum_{l_{k-1} \neq i,j,l_1,l_2,\ldots l_{k-2}} w_{i,l_1} w_{l_1,l_2} \ldots w_{l_{k-1},j} \end{eqnarray*}
 
Writing $\mathbf{P}=\mathbf{W}^T$ and interpreting $\mathbf{P}$ as a transition probability matrix for the DTMC $\{X_k\}$, obtained by reversing the edges in the social network, we have,
 
\begin{eqnarray*}
\lefteqn{g_{j}^{(\mathcal{N},i)}(k)}\\
&=& \sum_{l_1 \neq i,j} \sum_{l_2 \neq i,j,l_1} \ldots \sum_{l_{k-1} \neq i,j,l_1,l_2,\ldots l_{k-2}} p_{j,l_1} p_{l_1,l_2} \ldots p_{l_{k-1},i}\\
&=& \mathbb{P}(\{X_m \in \mathcal{N}\backslash \{i\}, 0 \leq m < k, X_k = i,\\
& &\hspace{3cm} X_0 \neq X_1 \neq \cdots \neq X_k \}\big| X_0=j)\\
&=:& c_k(j\rightarrow i)\\
\end{eqnarray*}
and similarly,

\begin{eqnarray*}
\lefteqn{g_{j}^{(\mathcal{N},\mathcal{A}_0)}(k)}\\
&=& \mathbb{P}(\{X_m \in \mathcal{N}\backslash \mathcal{A}_0, 0 \leq m < k, X_k \in \mathcal{A}_0,\\
& &\hspace{3cm} X_0 \neq X_1 \neq \cdots \neq X_k \} \big| X_0=j)\\
&=:& c_k(j\rightarrow \mathcal{A}_0)\\
\end{eqnarray*}
Define,
\[c(j \rightarrow i) = \sum_{k=0}^{\infty} c_k(j\rightarrow i)\]
\[c(j \rightarrow \mathcal{A}_0) = \sum_{k=0}^{\infty} c_k(j\rightarrow \mathcal{A}_0)\]
In the DTMC represented by $P$, given we start from state $j$, $c(j \rightarrow i)$ denotes the probability of hitting state $i$ for the first time through an acyclic path from $j$, and $c(j \rightarrow \mathcal{A}_0)$ denotes the probability of hitting the set $\mathcal{A}_0$ for the first time through an acyclic path from $j$. 
Since $g_{j}^{(\mathcal{N},\mathcal{A}_0)} = c(j \rightarrow \mathcal{A}_0)$ we have,

\begin{equation}
\sigma^{(\mathcal{N},\mathcal{A}_0)} = |\mathcal{A}_0| + \sum_{j \notin \mathcal{A}_0} c(j \rightarrow \mathcal{A}_0) 
\label{eq-sigma-in-c}
\end{equation}

Now the influence maximization problem can be restated as follows:\newline

\textit{ Let $\mathbf{W}$ denote the influence matrix for the given problem. For the DTMC over the finite state space $\mathcal{N}$, with transition probability matrix $\mathbf{P}=\mathbf{W}^T$, choose $\mathcal{A} \subset \mathcal{N}, |\mathcal{A}|=K$, such that $\sum_{j \in \mathcal{A}^c} c(j \rightarrow \mathcal{A})$ is maximized. The set $\mathcal{A}$ thus obtained is the solution to the original influence maximization problem.}

\subsection{Properties of Acyclic path probabilities}
\label{prop-app}

We shall use the interpretation in terms of acyclic path probabilities in the DTMC, to provide an alternate proof of submodularity of $\sigma^{(\mathcal{N}, \mathcal{A}_0)}$ in the next section. In this subsection, we state and prove a few properties of such acyclic path probabilities. Let us first introduce a more elaborate notation.

\begin{eqnarray}
\label{eq-defn}
\lefteqn{c^\mathcal{W} (j \stackrel{v}{\rightarrow} \mathcal{D})}\nonumber\\
&=&\mathbb{P}\bigg(\bigcup_{k=0}^{\infty}  \{X_m \in \mathcal{W}\backslash \mathcal{D}, 0 \leq m < k, X_k \in \mathcal{D},\nonumber\\
& &\hspace*{1cm} X_0 \neq X_1 \neq \cdots \neq X_k,\nonumber\\
& &\hspace*{1cm} \exists\ \ 0\leq u < k, s.t.\  X_u = v \} \big| X_0 = j\bigg)\nonumber\\ 
&=&\sum_{k=0}^{\infty}\mathbb{P}\bigg( \{X_m \in \mathcal{W} \backslash \mathcal{D}, 0 \leq m < k, X_k \in \mathcal{D},\nonumber\\
& &\hspace*{1cm} X_0 \neq X_1 \neq \cdots \neq X_k,\nonumber\\
& &\hspace*{1cm} \exists\ \ 0\leq u < k, s.t.\  X_u = v \} \big| X_0 = j\bigg) 
\end{eqnarray}
Here $c^\mathcal{W} (j \stackrel{v}{\rightarrow} \mathcal{D})$ is the probability in the DTMC of reaching the set $\mathcal{D}$ for the first time, through an acyclic path via node $v$, using nodes only from $\mathcal{W}$ as intermediate nodes, given that we start from node $j$. If $\mathcal{W}=\mathcal{N}$, we do not explicitly mention it in the notation. Also, $v$ is an optional argument, which when omitted, removes the constraint that $X_u=v$ for some $0\leq u \leq k$. In all useful cases that we consider, we assume $v,j \notin \mathcal{D}$ and also $v,j \in \mathcal{W}$.

Some properties of Acyclic path probabilities are as follows:

\begin{lemma}\label{lemma-jinA-1}
$c( j \rightarrow \mathcal{A} ) = 1$, for all $j \in \mathcal{A}$
\end{lemma}

\begin{proof}
\begin{eqnarray*}
c(j\rightarrow A)&=&\sum_{k=0}^{\infty}\mathbb{P}\bigg( \{X_m \in \mathcal{N} \backslash \mathcal{A}, 0 \leq m < k\},\\
& &\hspace*{1cm}X_k \in \mathcal{A}, X_0 \neq \cdots \neq X_k \big| X_0 = j\bigg)\\
&=& \mathbb{P}(X_0 \in \mathcal{A} \big| X_0 = j) +\\ 
& &\sum_{k=1}^{\infty}\mathbb{P}\bigg( X_0 \in \mathcal{N}\backslash \mathcal{A}, \{X_m \in \mathcal{N} \backslash \mathcal{A}, 0 < m < k\},\\
& &\hspace*{1cm} X_k \in \mathcal{A}, X_0 \neq X_1 \neq \cdots X_k \big| X_0 = j\bigg)
\end{eqnarray*}
Since $j\in \mathcal{A}$, all the terms in the summation are zero, and hence,
$c(j\rightarrow A) = \mathbb{P}(X_0 \in \mathcal{A} \big| X_0 = j) = 1$
\end{proof}

\begin{lemma}\label{lemma-jtoA-split}
$c(j\rightarrow \mathcal{A}) = c(j \stackrel{v}{\rightarrow} \mathcal{A}) + c^{\mathcal{N}\backslash \{v\}} (j \rightarrow \mathcal{A})$
\end{lemma}
This property means that the probability of reaching $\mathcal{A}$ starting from $j$ through an acyclic path, can be split into the probability of those paths via node $v$ and those that avoid node $v$. \newline

\begin{proof}
\begin{eqnarray*}
\lefteqn{c(j\rightarrow A)}\\
&=&\sum_{k=0}^{\infty}\mathbb{P}\bigg( \{X_m \in \mathcal{N} \backslash \mathcal{A}, 0 \leq m < k, X_k \in \mathcal{A},\\
& &\hspace{2cm}X_0 \neq X_1 \neq \cdots \neq X_k \}\big| X_0 = j\bigg)\\
&=&\sum_{k=0}^{\infty}\mathbb{P}\bigg( \{X_m \in \mathcal{N} \backslash\mathcal{A},0 \leq m < k, X_k \in \mathcal{A},\\
& &\hspace{2cm}X_0 \neq X_1 \neq \cdots \neq X_k ,\\   
& &\hspace{2cm} \exists\ 0\leq u < k,\ s.t.\ X_u =v \} \big| X_0 = j\bigg) + \\
& &\sum_{k=0}^{\infty}\mathbb{P}\bigg( \{X_m \in \mathcal{N} \backslash \mathcal{A},0 \leq m < k,\\
& &\hspace{2cm}X_k \in \mathcal{A}, X_0 \neq X_1 \neq \cdots\neq X_k ,\\ 
& &\hspace{2cm} \nexists\ 0\leq u < k,\ s.t.\ X_u =v \}\big| X_0 = j\bigg)\\
&=& c(j\stackrel{v}{\rightarrow}\mathcal{A}) +\\
& &\sum_{k=0}^{\infty}\mathbb{P}\bigg( \{X_m \in (\mathcal{N}\backslash\{v\})\backslash \mathcal{A},0 \leq m < k, X_k \in \mathcal{A}\backslash\{v\}, \\
& &\hspace*{2cm}X_0 \neq X_1 \neq \cdots \neq X_k \}\big| X_0 = j\bigg)\\
\end{eqnarray*}
Since, $v \notin \mathcal{A}$, $\mathcal{A}\backslash\{v\} = \mathcal{A}$ and we get,\newline
$c(j\rightarrow \mathcal{A}) = c(j \stackrel{v}{\rightarrow} \mathcal{A}) + c^{\mathcal{N}\backslash \{v\}} (j \rightarrow \mathcal{A})$
\end{proof}

\begin{lemma}\label{lemma-jtoAviav-split}
$c(j\rightarrow \mathcal{A} \cup \{v\}) = c^{\mathcal{N}\backslash \mathcal{A}} (j \rightarrow v) + c^{\mathcal{N}\backslash \{v\}} (j \rightarrow \mathcal{A})$, for all $v \notin \mathcal{A}$ \newline
\end{lemma}
This property means that the probability of reaching $\mathcal{A} \cup \{v\}$ through acyclic path, can be split into the probability of reaching $\mathcal{A}$ avoiding node $v$, and that of reaching node $v$, avoiding set $\mathcal{A}$. \newline
\begin{proof}
\begin{eqnarray*}
\lefteqn{c(j\rightarrow\mathcal{A}\cup \{v\})}\\
&=&\sum_{k=0}^{\infty}\mathbb{P}\bigg( \{X_m \in \mathcal{N} \backslash(\mathcal{A}\cup \{v\}), 0 \leq m < k, X_k \in \mathcal{A}\cup\{v\},\\
& &\hspace*{2cm} X_0 \neq X_1 \neq \cdots \neq X_k \}\big| X_0 = j\bigg)\\
&=& \sum_{k=0}^{\infty}\mathbb{P}\bigg( \{X_m \in (\mathcal{N}\backslash\mathcal{A})\backslash\{v\}, 0 \leq m < k, X_k = v,\\
& &\hspace*{2cm} X_0 \neq X_1 \neq \cdots \neq X_k \}\big| X_0 = j\bigg)+\\
& &\sum_{k=0}^{\infty}\mathbb{P}\bigg( \{X_m \in (\mathcal{N}\backslash\{v\})\backslash\mathcal{A}, 0 \leq m < k, X_k \in \mathcal{A},\\
& &\hspace*{2cm} X_0 \neq X_1 \neq \cdots \neq X_k \}\big| X_0 = j\bigg)\\
&=& c^{\mathcal{N}\backslash \mathcal{A}} (j \rightarrow v) + c^{\mathcal{N}\backslash \{v\}} (j \rightarrow \mathcal{A})\\ 
\end{eqnarray*}

The second equality results from the assumption that $v \notin \mathcal{A}$. Note that this property can be extended to $c(j\rightarrow \mathcal{A} \cup \mathcal{B})$, where  $\mathcal{A}$ and $\mathcal{B}$ are any two disjoint sets.
\end{proof}

Before stating and proving the next property, we wish to prove the following two sublemmas. 

\begin{sublemma}\label{sublemma-jtov}
$c(j \rightarrow v) = \sum_{L' \subseteq \mathcal{N}-\{j,v\}} \sum_{L \in \Pi(L')} p_{j,v} (L)$

where $\Pi(L')$ is the set of all permutations of $L'$, and for $L =\{ l_1, l_2, \cdots l_{k-1} \}$,
\begin{equation}
p_{j,v}(L) = \left\{ \begin{array}{ll}
            p_{j,l_1} p_{l_1,l_2} \cdots p_{l_{k-1},v} & \mbox{if $L = \{l_1, \ldots l_{k-1}\}$} \\
			p_{j,v} & \mbox{if $L = \emptyset$}
			\end{array}
			\right. \label{eq-pjv-pathL}
\end{equation}
\end{sublemma}
\begin{proof}
\begin{eqnarray*}
\lefteqn{c(j \rightarrow v)}\\
&=& \sum_{k=1}^{\infty} \mathbb{P} \bigg( X_1 =l_1,X_2=l_2,\ldots X_{k-1}=l_{k-1},X_k=v,\\
& & l_1 \neq l_2 \neq \ldots l_{k-1} \neq j \neq v \big| X_0 =j\bigg) \\
&=& \sum_{L' \subseteq \mathcal{N} -\{j,v\}} \sum_{L \in \Pi(L')} p_{j,v} (L)
\end{eqnarray*}

The second equality is by using the chain rule and the Markov property of $X_k$.
\end{proof}

\begin{sublemma}\label{sublemma-jtoAviav}

\[c(j \stackrel{v}{\rightarrow} \mathcal{A}) = \sum_{L' \subseteq \mathcal{N}\backslash\mathcal{A}-\{j,v\}} \sum_{L \in \Pi(L')} p_{j,v} (L) c^{\mathcal{N}\backslash(L'\cup\{j\})} (v \rightarrow \mathcal{A})\]
\end{sublemma}

\begin{proof}
\begin{eqnarray*}
\lefteqn{c(j\stackrel{v}{\rightarrow} \mathcal{A})}\\
&=& \sum_{k=0}^{\infty} \mathbb{P} \bigg( \{ X_m \in \mathcal{N}\backslash \mathcal{A} , 0 \leq m < k, X_k \in \mathcal{A},\\
& & \hspace{1cm} X_0 \neq X_1 \neq \cdots \neq X_k, \\
& & \hspace{1cm} \exists\ 0 \leq u < k, \ s.t.\ X_u = v \}\big| X_0=j \bigg) \\
&=& \sum_{k=0}^{\infty} \sum_{L' \subseteq \mathcal{N}\backslash\mathcal{A}} \sum_{u=1}^{k-1} \mathbb{P} \bigg( X_1 =l_1, \ldots ,X_{u-1} = l_{u-1}, X_u = v,\\ 
& & \hspace{1cm} l_i \notin \mathcal{A}, \{X_m \in \mathcal{N} \backslash \mathcal{A}, u < m < k\}, X_k \in \mathcal{A},\\
& &\hspace{1cm} X_0 \neq X_1 \neq \cdots \neq X_k \big| X_0 = j\bigg)\\
&=& \sum_{k=0}^{\infty} \sum_{L' \subseteq \mathcal{N}\backslash\mathcal{A}} \sum_{u=1}^{k-1} \mathbb{P} \bigg( X_1 =l_1, \ldots ,X_{u-1} = l_{u-1}, X_u = v,\\
& & l_1 \neq \ldots \neq l_{u-1}, l_i \neq v \neq j,\ l_i \notin \mathcal{A},\\
& &\{X_m \in (\mathcal{N} \backslash \{L' \cup j\})\backslash \mathcal{A}, u \leq m < k\}, X_k \in \mathcal{A} \big| X_0 =j\bigg)\\
&=&\sum_{k=0}^{\infty} \sum_{L' \subseteq \mathcal{N}\backslash\mathcal{A}} \sum_{u=1}^{k-1} \mathbb{P} \bigg( X_1 =l_1, \ldots ,X_{u-1} = l_{u-1},X_u = v,\\
& &\hspace{1cm} l_1 \neq \ldots \neq l_{u-1}, l_i \neq v \neq j, l_i \notin \mathcal{A} \big| X_0 =j\bigg)\times\\
& &\mathbb{P} \bigg(\{X_m \in (\mathcal{N} \backslash \{L' \cup j\})\backslash \mathcal{A}, u \leq m < k\},\\
& &\hspace{5cm}X_k \in \mathcal{A} \big| X_u =v\bigg)\\
&=& \sum_{L' \subseteq \mathcal{N} -(\mathcal{A} \cup \{j,v\})} \sum_{L \in \Pi(L')} p_{j,v} (L) c^{\mathcal{N} \backslash \{L' \cup j\}} (v\rightarrow \mathcal{A})
\end{eqnarray*}

where $p_{j,v}(L)$ defined as above in Equation~\ref{eq-pjv-pathL} and $L' = \{l_1, \ldots l_{u-1}\}$. The penultimate equality is by applying Markov property at $X_u =v$.
\end{proof}

Now we can state and prove the next property. 
\begin{lemma}\label{lemma-inequality}
$c(j \stackrel{v}{\rightarrow} \mathcal{A}) \leq c^{\mathcal{N} \backslash \mathcal{A}} ( j \rightarrow v) $ 
\end{lemma}
 
\begin{proof}
\begin{eqnarray*}
\lefteqn{c(j \stackrel{v}{\rightarrow} \mathcal{A})}\\
&=& \sum_{L' \subseteq \mathcal{N} -(\mathcal{A} \cup \{j,v\})} \sum_{L \in \Pi(L')}  p_{j,v} (L) c^{\mathcal{N} \backslash (L' \cup \{j\})} (v \rightarrow \mathcal{A})\\
&\leq& \sum_{L' \subseteq \mathcal{N} -(\mathcal{A} \cup \{j,v\})} \sum_{L \in \Pi(L')} p_{j,v} (L) \\
&=& c^{\mathcal{N} \backslash \mathcal{A}} (j \rightarrow v)
\end{eqnarray*}

In the above proof the equalities are due to Sublemmas~\ref{sublemma-jtov} and \ref{sublemma-jtoAviav}. The inequality is because $c^{\mathcal{N} \backslash (L' \cup \{j\})} (v \rightarrow \mathcal{A})$, being a probability term, is less than 1.
\end{proof}

\begin{lemma}\label{lemma-decrease-taboo}
For $\mathcal{A} \subseteq \mathcal{B}$, $c^{\mathcal{N} \backslash \mathcal{A}} (j \rightarrow v) \geq  c^{\mathcal{N} \backslash \mathcal{B}} (j \rightarrow v)$
\end{lemma}

\begin{proof}
Note that, \newline
\begin{equation}
c^{\mathcal{N} \backslash \mathcal{A}} ( j \rightarrow v) = \sum_{L' \subseteq \mathcal{N} -(\mathcal{A} \cup \{j,v\})} \sum_{L \in \Pi(L')}  p_{j,v} (L)
\label{eq-dec-paths}
\end{equation}

In the expression on the right, as $\mathcal{A}$ increases, the number of possible $L'$ decreases, hence $c^{\mathcal{N} \backslash \mathcal{A}}( j \rightarrow v)$ decreases.
\end{proof}

\section{Submodularity of $\sigma^{(\mathcal{N}, \mathcal{A}_0)}$}\label{sec-submod-monotone}

In this section we shall prove the submodularity and monotonicity of acyclic path probabilities and using them, we prove the monotonicity and submodularity of $\sigma^{(\mathcal{N}, \mathcal{A}_0)}$.

\begin{lemma}\label{lemma-c-monotonic}
$c ( j \rightarrow \mathcal{A} )$ is monotonically increasing in $\mathcal{A}$. i.e., For $\mathcal{A} \subseteq \mathcal{B}$, $c(j \rightarrow \mathcal{A}) \leq c(j \rightarrow \mathcal{B})$
\end{lemma}

\begin{proof}
We need to check, 
\[ c(j\rightarrow \mathcal{A} \cup \{v\}) - c(j\rightarrow \mathcal{A}) \stackrel{?}{\geq} 0\]

Substituting for $c(j\rightarrow \mathcal{A} \cup \{v\})$ and $c(j\rightarrow \mathcal{A})$ from Lemmas~\ref{lemma-jtoA-split} and \ref{lemma-jtoAviav-split}, we have

\begin{eqnarray*}
c(j\rightarrow \mathcal{A} \cup \{v\}) - c(j\rightarrow \mathcal{A}) &=& c^{\mathcal{N}\backslash \mathcal{A}} (j \rightarrow v) - c(j \stackrel{v}{\rightarrow} \mathcal{A})\\
&\geq& 0
\end{eqnarray*}

The last inequality results from Lemma~\ref{lemma-inequality}.
\end{proof}

\begin{lemma}\label{lemma-c-submodular}
$c ( j \rightarrow \mathcal{A} )$ is submodular in $\mathcal{A}$, i.e., $c ( j \rightarrow \mathcal{A}\cup\{v\} ) - c ( j \rightarrow \mathcal{A} )$ decreases, as  $\mathcal{A}$ increases.
\end{lemma}

\begin{proof}
\begin{eqnarray*}
\lefteqn{c(j\rightarrow \mathcal{A} \cup \{v\}) - c(j\rightarrow \mathcal{A}) = c^{\mathcal{N}\backslash \mathcal{A}} (j \rightarrow v) - c(j \stackrel{v}{\rightarrow} \mathcal{A})}\\
&=& \sum_{L' \subseteq \mathcal{N} -(\mathcal{A} \cup \{j,v\})} \sum_{L \in \Pi(L')} p_{j,v} (L) - \\
& &\sum_{L' \subseteq \mathcal{N} -(\mathcal{A} \cup \{j,v\})} \sum_{L \in \Pi(L')} p_{j,v} (L) c^{\mathcal{N} \backslash (L' \cup \{j\})} (v \rightarrow \mathcal{A}) \\
&=& \sum_{L' \subseteq \mathcal{N} -(\mathcal{A} \cup \{j,v\})} \sum_{L \in \Pi(L')} p_{j,v} (L) \bigg[ 1 - c^{\mathcal{N} \backslash (L' \cup \{j\})} (v \rightarrow \mathcal{A})\bigg]
\end{eqnarray*}
For a fixed $L'$, as $\mathcal{A}$ increases, $c^{\mathcal{N} \backslash (L' \cup \{j\})} (v \rightarrow \mathcal{A})$ increases by Lemma~\ref{lemma-c-monotonic}, and hence the entire term inside the summation decreases, as $\mathcal{A}$ increases. Also, as $\mathcal{A}$ increases, the number of $L'$ that satisfy the constraint in the first summation also decrease. Hence, $c(j\rightarrow \mathcal{A} \cup \{v\}) - c(j\rightarrow \mathcal{A})$ decreases, as $\mathcal{A}$ increases. Thus $c(j\rightarrow \mathcal{A})$ is submodular.
\end{proof}

\subsection{Monotonicity and Submodularity of $\sigma^{(\mathcal{N},\mathcal{A}_0)}$}

Recalling  Equation~\ref{eq-sigma-in-c},

\[\sigma^{(\mathcal{N},\mathcal{A}_0)} = |\mathcal{A}_0| + \sum_{j \notin \mathcal{A}_0} c(j \rightarrow \mathcal{A}_0) \]

Since $c( j \rightarrow \mathcal{A} ) = 1$, for all $j \in \mathcal{A}$ by Lemma~\ref{lemma-jinA-1} we can write

\[\sigma^{(\mathcal{N},\mathcal{A}_0)} = \sum_{j \in \mathcal{N}} c(j \rightarrow \mathcal{A}_0) \]

$c(j \rightarrow \mathcal{A}_0)$ is monotonically increasing and submodular in $A_0$ by Lemmas~\ref{lemma-c-monotonic} and \ref{lemma-c-submodular} and since $\sigma^{(\mathcal{N},\mathcal{A}_0)}$ is a non-negative linear combination of such $c(j \rightarrow \mathcal{A}_0)$,  this automatically proves the monotonicity and submodularity of $\sigma^{(\mathcal{N},\mathcal{A}_0)}$.
 
\section{Examples} 
\label{examples}
In this section, we use the analytical expression to obtain the optimal initial set for some simple LT models. We also show that PageRank matches with the optimal solution obtained from the analytical expression in 2 cases. We provide simulation results for those cases in which PageRank is not optimal, but provides a very good approximation of the greedy solution.
 
\subsection{UISLT Models on a Complete Graph}
We introduce a simple version of the Linear Threshold model, called the Uniform Influence-Susceptance Linear Threshold model (UISLT). In this model, we have two parameters $\alpha_{i}$ and $\beta_{i}$ associated with the node $i$, a measure of the level of influence and susceptance of the node $i$. The social network is a complete graph with the matrix $\mathbf{W}$ defined as follows:\newline
For all $i,j$ with $i \neq j$,  
\[w_{i,j} = \alpha_{i} \times \beta_{j}, \mbox{for all} j \neq i\]
and 

\[w_{i,i} = 1- \sum_{j \neq i} w_{j,i} \]
Note that $\alpha_{i}$'s and $\beta_{i}$'s are chosen such that, $\sum_{j \neq i} w_{j,i} \leq 1$. This implies that, for all $i$,

\[\sum_{j\neq i} \alpha_j \leq \frac{1}{\beta_i} \]

\begin{theorem}
Let $\mathcal{N}=\mathcal{A}_0\cup \{j_1,j_2,\cdots j_k \}$. where $K =|\mathcal{A}_0|$ and  $k = |\mathcal{N}| - K$. Then the total influence of the initial set $\mathcal{A}_0$ in the UISLT model is given by

\[ \sigma^{(\mathcal{N},\mathcal{A}_0)} = |\mathcal{A}_0| + \alpha_{\mathcal{A}_0} \sum_{m=0}^{k-1} h^{(m)} (\mathbf{\alpha},\mathbf{\beta}) \] 

where $\mathbf{\alpha} = \{\alpha_{j_1},\alpha_{j_2},\cdots \alpha_{j_k} \}$ ,\ $\mathbf{\beta} = \{\beta_{j_1},\beta_{j_2},\cdots \beta_{j_k} \}$ ,\ $\alpha_{\mathcal{A}_0} = \sum_{i \in \mathcal{A}_0} \alpha_{i}$, and

\begin{eqnarray*}
\lefteqn{h^{m} (\mathbf{\alpha},\mathbf{\beta})}\\ 
&=& \sum_{ \{l_1,\ldots,l_{m+1}\} \subseteq \{1,\ldots,k\} } \beta_{j_{l_1}} \times \cdots \times \beta_{j_{l_{m+1}}} f^{m}(\alpha_{j_{l_1}},\ldots,\alpha_{j_{l_{m+1}}})\\
\end{eqnarray*}

where,

\[f^{0}(x_1,\ldots,x_t) = 1 \]

and for all $m > 0$,

\[f^{m}(x_1,\ldots,x_t) = m!\sum_{\{y_1,\ldots,y_m\} \subseteq \{x_1,\ldots,x_t\} } y_1 \times \cdots \times y_m \]

\end{theorem}

\begin{proof}
The proof is by direct application of Equations~\ref{eq-sigma-i} and \ref{eq-sigma-a0}. 
\end{proof}

\paragraph{USLT model}
For the Uniform Susceptance model, from the above general equation for UISLT, by setting $\mathbf{\alpha}  = (1,\ldots,1)$, and noting that,

\[f^{m}(\alpha_{j_{l_1}},\ldots,\alpha_{j_{l_{m+1}}}) = (m+1)! \]

\[ \alpha_{\mathcal{A}_0} = |\mathcal{A}_0| \]
we get, 

\[ \sigma^{(\mathcal{N},\mathcal{A}_0)} = |\mathcal{A}_0| + |\mathcal{A}_0| \sum_{m=0}^{k-1} f^{m+1} (\mathbf{\beta}) \] 
where $f^{m}$ is the same as defined earlier. 

Hence in this case $\sigma^{(\mathcal{N},\mathcal{A}_0)}$ is an increasing function of $\mathbf{\beta}$. Thus to maximize $\sigma^{(\mathcal{N},\mathcal{A}_0)}$, we need to pick $\mathcal{A}_0$ such that, the nodes with maximum $\beta_i$ are left out. Thus the optimal $\mathcal{A}_0$ in this case, is the set of $K$ nodes with least $\beta_i$ values. This also makes intuitive sense, since by picking this $\mathcal{A}_0$ (the least susceptible nodes), we ensure that the inactive nodes are the most susceptible ones, and hence maximizing expected cardinality of the terminal set. 

It turns out that when we apply the PageRank algorithm, we get the stationary probability as, 

\[ \pi_{i} = \frac{1/\beta_i}{\sum_{j} 1/\beta_j} \] 

Thus choosing the nodes with top-k $\pi_i$ yields us the same optimal set, i.e., the set of $K$ nodes with least $\beta_i$ values.
    
\paragraph{UILT model}
For the Uniform Influence model, from the above general equation for UISLT, by setting $\mathbf{\alpha}  = (1,\ldots,1)$, we get,

\begin{eqnarray*}
\lefteqn{h^{m} (\mathbf{\alpha})}\\
&=&\sum_{ \{l_1,\ldots,l_{m+1}\} \subseteq \{1,\ldots,k\} } f^{m}(\alpha_{j_{l_1}},\ldots,\alpha_{j_{l_{m+1}}})\\
\end{eqnarray*}

Hence, \[h^{m}(\mathbf{\alpha}) = (k-m) f^m (\alpha_{j_1}, \ldots, \alpha_{j_k} ) \] 

\[ \sigma^{(\mathcal{N},\mathcal{A}_0)} = |\mathcal{A}_0| + k \alpha_{\mathcal{A}_0} + \alpha_{\mathcal{A}_0} \sum_{m=1}^{k-1} (k-m) f^m (\alpha_{j_1}, \ldots, \alpha_{j_k} ) \] 

Hence in this case $\sigma^{(\mathcal{N},\mathcal{A}_0)}$ is an increasing function of $\alpha_{\mathcal{A}_0}$, fixing $|\mathcal{A}_0|$ to be $K$. Thus to maximize $\sigma^{(\mathcal{N},\mathcal{A}_0)}$, we need to pick $\mathcal{A}_0$ such that $\alpha_{\mathcal{A}_0}$ is maximized. Thus the optimal $\mathcal{A}_0$ in this case, is the set of $K$ nodes with highest $\alpha_i$ values. This also makes intuitive sense, since $\alpha_i$ is a measure of influence, and the $\mathcal{A}_0$ thus obtained would be the set of most influential nodes.

It turns out that when we apply the PageRank algorithm, we get the stationary probability as, 

\[ \pi_{i} = \frac{\alpha_i}{\sum_{j}\alpha_j} \] 

Thus choosing the nodes with top-k $\pi_i$ yields us the same optimal set, i.e., the set of $K$ nodes with highest $\alpha_i$ values.

\subsubsection{UISLT model and PageRank}
\label{sec-uislt-pagerank}
But in general, the PageRank algorithm need not be optimal for the UISLT case. For PageRank, in the UISLT case, the stationary probability is given by,

\[ \pi_{i} = \frac{\alpha_i/\beta_i}{\sum_{j} \alpha_j/\beta_j} \] 

Hence one might suspect that picking the nodes in increasing order of $\alpha_i/\beta_i$ could be optimal. But it turns out to be false, since a node with $\beta_i$ very close to zero could get chosen as the most influential node irrespective of its $\alpha_i$.  If we restrict the $\beta_i$, by not allowing it vary much, we see that the PageRank algorithm gives a very good approximation of the greedy solution.

The following simulation was conducted on a complete graph with 50 nodes with the UISLT model. The $\alpha_i$'s were picked at random from a uniform distribution over $[0,1]$ and $\beta_i$'s were picked with a uniform distribution over $[\frac{0.5}{\sum_{j\neq i}\alpha_j},\frac{1}{\sum_{j\neq i}\alpha_j}]$. It is found that PageRank performs on par with the greedy algorithm. Results are shown in Figure~\ref{uislt_50}.

\begin{figure}[http]
\centerline{\includegraphics[scale=0.3]{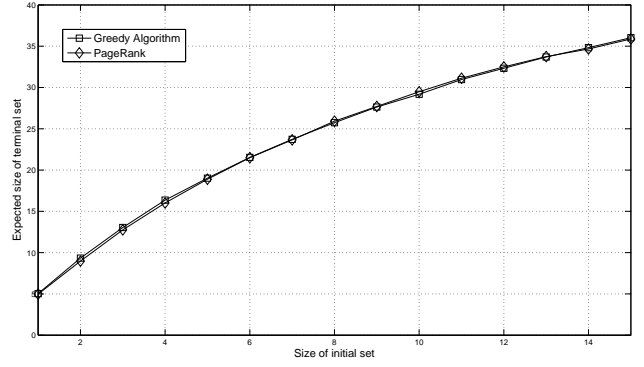}}
\caption{UISLT on a complete influence graph of 50 nodes, with $\alpha_i$'s and $\beta_i$'s being picked as described in Section~\ref{sec-uislt-pagerank}}
\label{uislt_50}
\end{figure}

\subsection{Node Degree based Model}
In this class of models, we start with an undirected graph without self-loops, whose adjacency matrix is given by $A$. We then generate the influence matrix $\mathbf{W}$ by normalizing the adjacency matrix as follows:

\begin{equation} 
w_{i,j} = a_{i,j} / d_{j} 
\label{deg-infl-eqn}
\end{equation}
where $d_{j} = \sum_{i} a_{i,j}$ is the degree of the node $j$.

Let us restrict our attention to \emph{acyclic graphs}. We then have the following theorem.

\begin{theorem}
Consider an acyclic undirected graph $\mathcal{N}$ represented by the adjacency matrix $A$. Let the influence matrix be generated by Equation \ref{deg-infl-eqn}. Then, for any node $i \in \mathcal{N}$,
\[ \sigma^{(\mathcal{N},i)} = d_{i} + 1 \]
\end{theorem}

\begin{proof}
Given the acyclic graph $\mathcal{N}$ and the node $i$, view the graph as a tree $\mathcal{T}$ of depth $D$, with node $i$ as the root. For any node $j$ in the tree $\mathcal{T}$, let $P(j)$ be the parent of node $j$ in $\mathcal{T}$, and $C(j)$ be its immediate child nodes. Define,

\[ L_{0} = \{ j \in \mathcal{T} : C(j) = \emptyset \}  \bigg( \neq \emptyset \bigg) \]
and for $0 < k \leq D$,

\[ L_{k} = \{ j \in \mathcal{T}: C(j) \subseteq \bigcup_{t=0}^{k-1} L_t, C(j)\cap L_{k-1} \neq \emptyset \} \]

Hence by definition of depth $D$ we have, $L_D =\{i\}$ and it is easy to see that $L_k$'s partition nodes in $\mathcal{N}$ into sets of nodes having the same depth.

By Equation~\ref{eq-sigma-i} we have,

\begin{equation}
\sigma^{(\mathcal{N}, i)} = 1 +  \sum_{j\in C(i)} \frac{1}{|C(j)|+1} \sigma^{(\mathcal{N}\backslash P(j), j)} 
\label{infl-child-eqn}
\end{equation}
where $P(j) =i$ and $j\in L_0 \cup \cdots \cup L_{D-1}$.

We shall prove inductively that, $\forall j$,
\begin{equation}
\sigma^{(\mathcal{N}\backslash P(j), j)} = |C(j)| + 1 
\label{claim-deg-infl}
\end{equation}
We know that if $j \in L_{0}$, then it is true, since in that case $\sigma^{(\mathcal{N}\backslash P(j), j)} = 1$ and $C(j) =0$. Assume that the claim is valid for $j \in L_{0}, L_{1}, \cdots, L_{k}$. 

For $j \in L_{k+1}$,

\begin{eqnarray*}
\sigma^{(\mathcal{N} \backslash P(j), j)} &=& 1 +  \sum_{l\in C(j)} \frac{1}{|C(l)|+1} \sigma^{(\mathcal{N}\backslash P(l), l)}\\
&=& 1 + \sum_{l \in C(j)}  \frac{|C(l)|+1}{|C(l)|+1} \\
&=& 1+ |C(j)|
\end{eqnarray*}
The second equality in the above set of equations, is because Claim~\ref{claim-deg-infl} is valid for $l \in \cup_{m=0}^{k} L_{m}$. Thus substituting the above in Equation~\ref{infl-child-eqn}, we have

\[ \sigma^{(\mathcal{N}, i)} = |C(i)| + 1 = d_i + 1 \]
\end{proof}
Thus it is found that the most influential node is the node with the highest degree. In order to pick the second node for the greedy algorithm we need to maximize $\sigma^{(\mathcal{N}, i\cup j)}$ where $i$ is the node with the highest degree. By using Equations~\ref{eq-sigma-i} and \ref{eq-sigma-a0} we can write,

\[ \sigma^{(\mathcal{N},i\cup j)} = \sigma^{(\mathcal{N}\backslash i,j)} + \sigma^{(\mathcal{N}\backslash j,i)} \]
\[ \sigma^{(\mathcal{N},i)} = \sigma^{(\mathcal{N}\backslash j,i)} + w_{i,j} \sigma^{(\mathcal{N}\backslash i,j)} \]
\[ \sigma^{(\mathcal{N},j)} = \sigma^{(\mathcal{N}\backslash i,j)} + w_{j,i} \sigma^{(\mathcal{N}\backslash j,i)} \]
From the above expressions, we find that if $i$ and $j$ are high degree nodes, their net influence can be approximated well by the sum of their individual influences, since $w_{i,j}$ and $w_{j,i}$ are small. Extending this further, we hence see that the solution of picking the high degree nodes will give us a very good approximation of the greedy solution.

By applying the PageRank algorithm using $P=\mathbf{W}^T$ as the transition probability matrix, where $\mathbf{W}$ is as defined in Equation~\ref{deg-infl-eqn} we get the stationary probability to be,

\[ \pi_i = \frac{d_i}{\sum_j d_j} \] 
and we find that PageRank algorithm matches with the heuristic of choosing the nodes according to degree and hence will give a good approximation of the greedy solution.

We also tested PageRank algorithm on undirected graphs which may have cycles. It turns out that there is still a high correlation (see Figure~\ref{deginflcorr}) between the degree of a node and its individual influence. Hence even in this case we can pick nodes in the decreasing order of degree to get a good estimate of the optimal initial set. 
\begin{figure}[http]
\centerline{\includegraphics[scale=0.3]{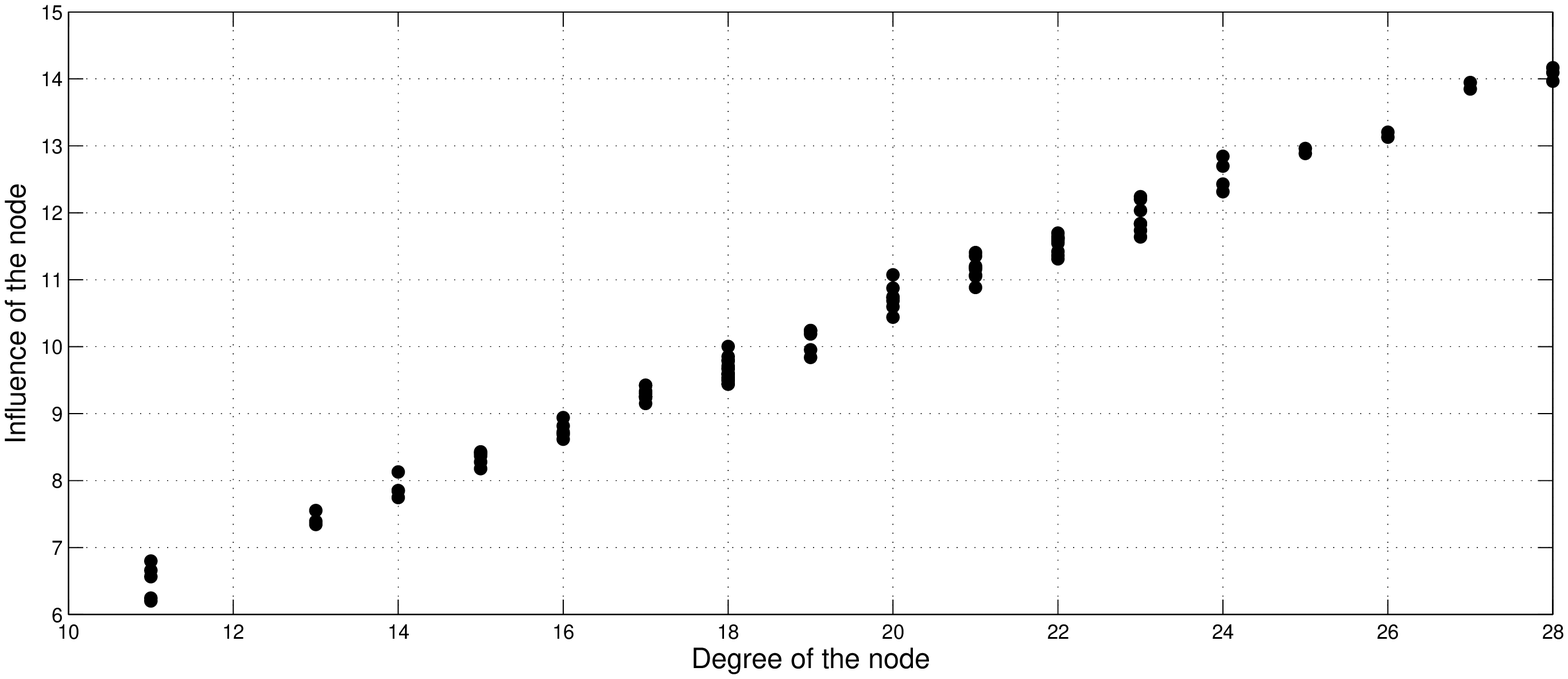}}
\caption{Scattergram between degree and influence of a node under the degree-based influence model in an Erdos-Renyi random graph of 100 nodes}
\label{deginflcorr}
\end{figure}

\section{Improving the Greedy Algorithm}
\label{sec-g1s}
As shown by Kempe et al.\ the greedy hill-climbing solution achieves $(1-1/e)$ approximate solution for the influence maximization problem \cite{kempe-etal03max-spread-infl}. But finding the initial set using the greedy algorithm is computationally quite expensive. There have been several efforts in the literature to improve the execution speed of the greedy algorithm. In this section, based on the insights obtained from the analytical expressions, we provide two techniques, namely thresholding and restriction, which achieve a very close approximation to the greedy solution.

 During the greedy algorithm, the first stage involves computation of $\sigma^{(\mathcal{N},i)}$ for all nodes $i\in \mathcal{N}$. One can rank the nodes in decreasing order of individual influences to yield a rank list which we shall refer to as \textbf{G1} list. One possible solution for influence maximization is to pick the top $K$ nodes to be the initial set. This solution is not as effective as the greedy solution, because it may contain \emph{dummy} nodes, i.e., nodes which in spite of having a high individual influence, fail to provide high marginal contributions to the greedy solution set and hence get rejected by the greedy algorithm. 
 
 We classify the \emph{dummy} nodes into two categories, namely the \emph{leechers} and \emph{subordinates}. Let $\mathcal{X}$ denote the set of nodes in \textbf{G1} which are above $i$ and have been picked to be part of the initial set. Then a node $i$ is a \emph{leecher} to set $\mathcal{X}$, if

\[ \sigma^{(\mathcal{N},i)} \approx \sum_{j \in \mathcal{X}} w_{i,j} \sigma^{(\mathcal{N}\backslash i,j)}\] 
 
This means that the node $i$ will have almost nil marginal contribution to a set which already contains $X$, since it primarily derives all its influence from those nodes.

We define a node $i$ to be an $\alpha$-subordinate of set $\mathcal{X}$ if 

\[ g_{i}^{(\mathcal{N}, \mathcal{X})} > \alpha \]

This means that node $i$ gets activated at least $\alpha$ fraction of the time when we begin with $\mathcal{X}$ as the initial set. Thus, for high enough $\alpha$, the marginal contribution of this node to the set $\mathcal{X}$ will be much smaller than its individual influence. If one can identify and eliminate such leechers and subordinates from \textbf{G1} while picking the initial set, then a more effective initial set could be obtained.

We use two techniques namely \emph{thresholding} and \emph{restriction} to filter out the subordinate nodes and leechers respectively. Thresholding involves comparing the  $g_{i}^{(\mathcal{N}, \mathcal{X})}$ with $\alpha$ and restriction involves evaluating $\sigma^{(\mathcal{N}\backslash \mathcal{X}, i)}$ to pick the next best node. One can also choose to use only one of those techniques, with slightly reduced effectiveness.

\subsection{G1-Sieving Algorithm}
In this algorithm, we first evaluate $\sigma^{(\mathcal{N},i)}$ for $i\in \mathcal{N}$ and obtain the \textbf{G1} list. We start with $\mathcal{X}=\{i\}$. We remove the nodes which are $\alpha$-subordinates of set $\mathcal{X}$ by evaluating $g_{i}^{(\mathcal{N}, \mathcal{X})}$ for all nodes and comparing with the threshold $\alpha$. For the remaining nodes, we compute $\sigma^{(\mathcal{N}\backslash \mathcal{X},i)}$ and pick the node that has the highest value. One can also discard the nodes which have a value very close to zero, since these nodes are the leechers. We add the picked node to $\mathcal{X}$ and repeat the procedure until we have a set of size $K$ or we exhaust the entire list. The algorithm is shown in Algorithm~\ref{algo-g1s}.  

\begin{algorithm}
\LinesNumbered
Evaluate $\sigma^{(\mathcal{N},i)}$ for all $i \in \mathcal{N}$\;
Sort nodes in decreasing order of $\sigma^{(\mathcal{N},i)}$ to get G1\;
$\mathcal{X}= G1(1)$\;
$G1 = G1 \backslash \mathcal{X}$\;
\For{$c = 2$ to $K$}{
  \If{$G1 = \emptyset$}{
    break\;
  }
  \For{$i \in G1$}{
  
    \If{$g_{i}^{(\mathcal{N},\mathcal{X})} > \alpha \ \| \ \sigma^{(\mathcal{N}\backslash \mathcal{X}, i)} < \epsilon$}{
	  $G1 = G1 \backslash \{i\}$\;
	  continue\;
	}
	Evaluate $\sigma^{(\mathcal{N}\backslash \mathcal{X}, i)}$\;
  }
  $v = \arg\max_{i \in G1} \sigma^{(\mathcal{N}\backslash \mathcal{X}, i)}$\;
  $\mathcal{X} = \mathcal{X}\cup v$\;
  $G1 = G1 \backslash \{v\}$\;
}
\caption{G1-Sieving Algorithm}
\label{algo-g1s}
\end{algorithm} 

\section{Simulations}

\subsection{Coauthorship Networks}
 
Newman \cite{newman01scientific-collaboration} observed that scientific colloboration networks are excellent examples of social networks. In such networks, each node represents an author in the scientific community under consideration, and an edge exists between two nodes $i$ and $j$ if those two authors are listed as co-authors at least in one of the papers. Newman in \cite{newman01scientific-collab-weightednetworks} explains a method by which the strength of collaboration (symmetric) between two authors can be extracted. We use this data to obtain the Linear Threshold model parameters. The process is as follows.

Let $\mathcal{R}$ denote the set of all papers under consideration in the scientific community, excluding the papers that only have a single author. For each $r \in \mathcal{R}$, let $n_r$ represent the number of coauthors for paper $r$. Let $\mathcal{N}$ represent the union of all authors of the papers in $\mathcal{R}$. Define $\delta(i,r)$ to be 1 if author $i$ was a co-author of paper $r$ and zero otherwise. Then ${\Tilde{w}}_{i,j}$ representing the strength of collaboration between authors $i$ and $j$, for $i\neq j$, is given by,

\[ {\Tilde{w}}_{i,j} = \sum_{r \in \mathcal{R}} \frac{\delta(i,r) \delta(j,r)}{n_r - 1} \]

We do not define terms of the form ${\Tilde{w}}_{i,j}$ since they do not represent any measure of collaboration. Now , by using these ${\Tilde{w}}_{i,j}$ 's, we obtain the entries of influence matrix $\mathbf{W}$ by normalizing. i.e.,

\[ w_{i,j} = \frac{{\Tilde{w}}_{i,j}}{\sum_{k=1,k\neq j}^{|\mathcal{N}|} {\Tilde{w}}_{i,j}} \]

The simulations in the following sections are carried out on a coauthorship network of NetScience community containing 1589 nodes.

\subsection{Comparing PageRank with the Greedy Algorithm}

In Section~\ref{examples} we examined various cases, where the PageRank algorithm was either optimal or performed very well. Here we report simulations comparing PageRank and Greedy algorithm on the Netscience dataset. They are also compared against heuristics such as the out-degree(obtained by counting the number of outgoing edges) and the weighted out-degree(summing up the weights on the outgoing edges).  In the scenario where the $w_{i,j}$'s are obtained as above, it turned out that the PageRank algorithm's performance was much below that of the greedy algorithm. We ran another set of simulations where $\mathbf{\Tilde{W}}$ was interpreted as directed, for experimental purposes. This means that the ``strength of collaboration'' ${\Tilde{w}}_{i,j}$ between two authors $i$ and $j$ was completely assigned to the author with a higher index, i.e., for $i < j$, ${\Hat{w}}_{j,i}$ = ${\Tilde{w}}_{i,j}$ and ${\Hat{w}}_{i,j} = 0$. This would result in an influence graph which is directed and has no cycles. It was found that PageRank algorithm performed on par with the greedy algorithm. The results are shown in figures~\ref{compare-undirected} and \ref{compare-directed}.

\begin{figure}[http]
\centerline{\includegraphics[scale=0.3]{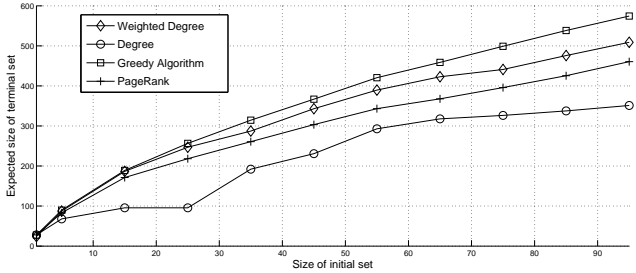}}
\caption{Comparison of PageRank and Greedy - undirected case}
\label{compare-undirected}
\end{figure}
 
\begin{figure}[http]
\centerline{\includegraphics[scale=0.3]{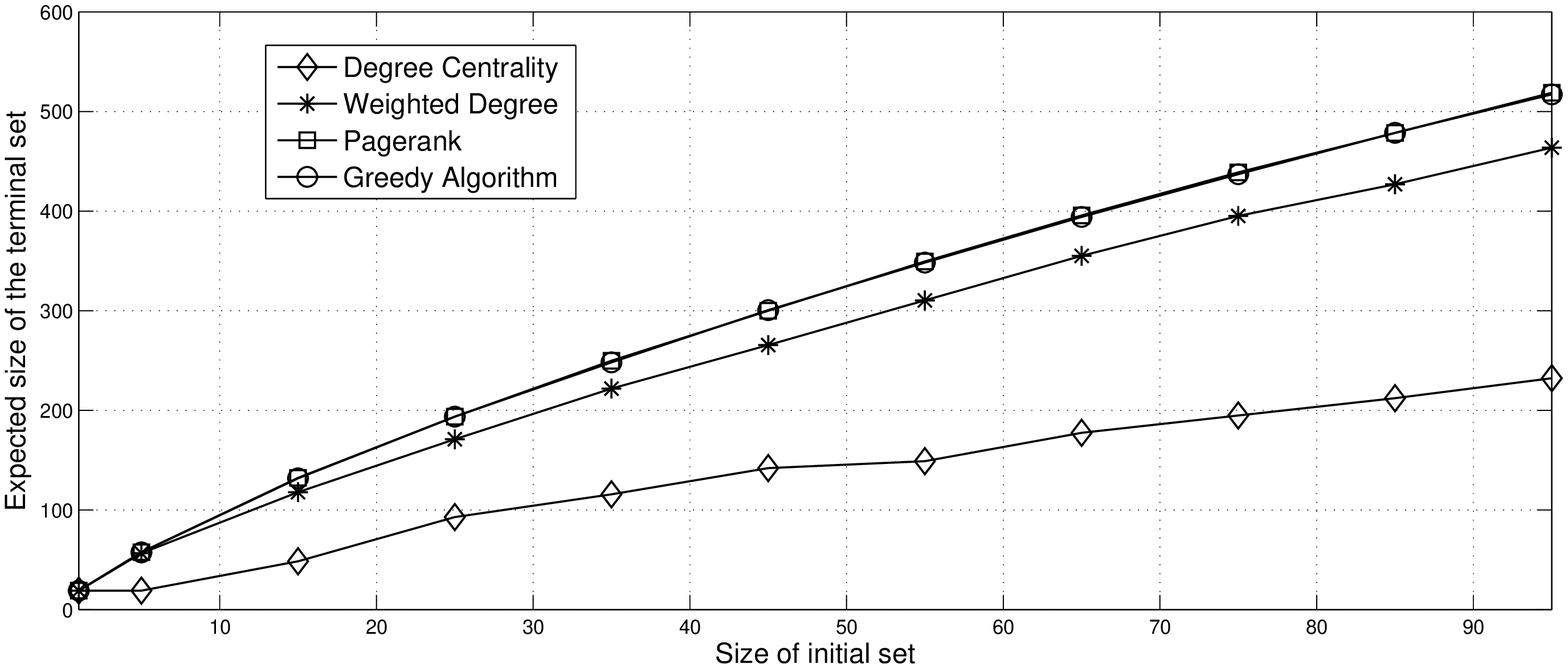}}
\caption{Comparison of PageRank and Greedy - directed case}
\label{compare-directed}
\end{figure}
 
We see that this is because, the PageRank algorithm essentially works with random walks on the given graph and finding the stationary probability. But as we have pointed out in Section~\ref{DTMC-interpret} maximizing the spread of influence involves random walks that are ``self-avoiding''. Thus it turns out that in the directed case, where there are no cycles, PageRank algorithm works fine, whereas in the undirected case, it fails to perform on par with the greedy algorithm. Thus, even though it can be calculated efficiently to give a rough estimate of the \emph{network effect} of a node, in many cases it might perform poorly compared to the greedy algorithm.

\subsection{Comparison of G1-Sieving with Greedy Algorithm}

We carried out two experiments with the NetScience dataset for the G1-sieving algorithm. The first involved using only the thresholding technique with various values for $\alpha$. It is interesting to note the change in performance of the thresholding technique with variation of $\alpha$ as shown in Figure~\ref{greedy_various_sieving}. For low values of $\alpha$, the algorithm retains only the nodes whose influence domains are almost disjoint, and hence performs badly. Also, for high values of $\alpha$, the algorithm does not remove many nodes, and the list almost resembles the \textbf{G1} list. It turned out that for this dataset, thresholding with $\alpha=0.3$ provided a very good approximation of the greedy solution. 

In the second case, we implement the G1-Sieving algorithm with the threshold $\alpha=0.3$ and with restriction. It is found that restriction provides a slight improvement over the solution with only thresholding. The results are shown in Figure~\ref{greedy_sieving_restriction}. We find that G1-Sieving algorithm performs on par with the greedy algorithm.\newline
 
\begin{figure}[http]
\centerline{\includegraphics[scale=0.3]{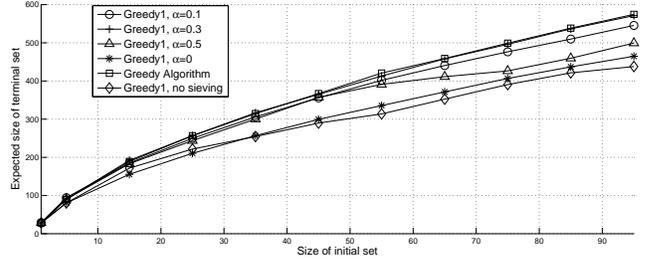}}
\caption{Comparison of various sieving thresholds with Greedy algorithm}
\label{greedy_various_sieving}
\end{figure}
 
\begin{figure}[http]
\centerline{\includegraphics[scale=0.3]{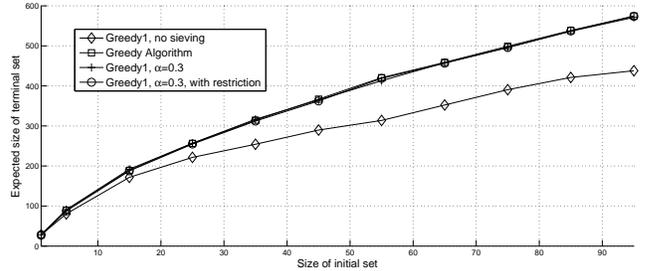}}
\caption{Comparison of G1-Sieving with Greedy algorithm}
\label{greedy_sieving_restriction}
\end{figure}

 When $\mathcal{A} \subseteq \mathcal{B} \subseteq \mathcal{N}$ and set sizes are very small compared to $|\mathcal{N}|$, due to the monotonicity of the influence function, evaluating $\sigma^{(\mathcal{N},\mathcal{A})}$ is faster than $\sigma^{(\mathcal{N},\mathcal{B})}$, since the latter involves more number of activations, whereas for set sizes closer to $|\mathcal{N}|$, evaluating $\sigma^{(\mathcal{N},\mathcal{B})}$ is faster than $\sigma^{(\mathcal{N},\mathcal{A})}$, since the former involves fewer nodes to be activated. Also for a given set size, the time taken for evaluating the influence of a set $\mathcal{A}$ increases with the expected influence of the set. 
 
 Keeping these in mind, we see that G1-Sieving algorithm runs much faster than greedy algorithm, since it evaluates influences for only sets with single node as the initial set, and as it proceeds it evaluates influences for fewer influential nodes in a restricted graph. Also, the technique of employing only thresholding, performs almost on par with the Greedy algorithm, given the fact that it involves evaluating $\sigma^{(\mathcal{N},i)}$ for all $i \in \mathcal{N}$, i.e., the first stage of greedy algorithm, and then subsequently evaluating influences (to get the activation probabilities $g_{j}^{(\mathcal{N},\mathcal{X})})$ only once per round. 

For the G1-Sieving algorithm, we obtained the optimal $\alpha$ using simulations. It would be interesting to look at ways to estimate $\alpha$ based on the graph structure. One can also use variable $\alpha$, by having higher threshold for initial nodes while reducing it for later stages.

\section{Discussion}

In this paper we have derived an analytical expression for the influence of a given set in a social network under the Linear threshold model. The insights thus obtained helped us propose a better algorithm for choosing the initial set to maximize the spread of influence. A similar approach could be adopted for the independent cascade model. This will help us explain why certain heuristics work well and also help in developing better algorithms. It is also to be noted that the current framework can be easily extended to the time constrained influence maximization problem, where the activation process is terminated after a fixed number of steps. Another interesting implication of this work is the role played by self avoiding random walks in the analytical expression for the influence function. Finding an efficient way to compute these probabilities will speed up the influence computations. PageRank algorithm was found to be sub-optimal since it was working on the assumption of random walks which could involve cycles. As an interesting aside, one can even model the walk of random surfer on the Web graph to be a ``self-avoiding'' random walk which can have some implications on the Web-page ranking algorithms. 

\bibliographystyle{unsrt}
\bibliography{bib-infl-nodes}

\end{document}